\documentclass[a4paper,12pt]{article}
\usepackage{amsmath,amsthm}
\usepackage{amssymb,amsmath}
\usepackage{amsfonts}
\usepackage[english]{babel}
\usepackage{hyperref}
\title{Pre-Hamiltonian operators}
\author{S. Startsev}
\date{}

\def\pa{\partial}                 
         
\def\be{\begin{equation}}       \def\ba{\begin{array}}
\def\ee{\end{equation}}         \def\ea{\end{array}}
\def\im    {\mathop{\rm Im} \nolimits}

\numberwithin{equation}{section}

\def\Z{\mathbb{Z}}

\let\ds\displaystyle

\theoremstyle{plain}
\newtheorem{theorem}{Theorem}[section]
\newtheorem{lemma}[theorem]{Lemma}
\newtheorem{corollary}[theorem]{Corollary}
\newtheorem{proposition}[theorem]{Proposition}

\theoremstyle{remark}
\newtheorem{remark}[theorem]{Remark}
\newtheorem{example}[theorem]{Example}

\theoremstyle{definition}
\newtheorem{definition}[theorem]{Definition}

\begin{document}

%%%%%%%%%%%%%%%%%%%%%%%%%%%%%%%%%%%%%%
\baselineskip=15pt
\vspace{1cm} \centerline{{\LARGE \textbf {Pre-Hamiltonian operators related to   
 }}}
\vspace{0.3cm} \centerline{{\LARGE \textbf {hyperbolic equations of Liouville type 
 }}}

\vskip1cm \hfill
\begin{minipage}{16.5cm}
\baselineskip=15pt {\bf S. Ya. Startsev}
\\ [2ex] {\footnotesize
Institute of Mathematics, Ufa Federal Research Centre, Russian Academy of Sciences  
\\}
\vskip1cm{\bf Abstract.} This text is devoted to hyperbolic equations admitting differential operators that map any function of one independent variable into a symmetry of the corresponding equation. We use the term `symmetry driver' for such operators and prove that any symmetry driver of the smallest order is pre-Hamiltonian (i.e., the image of the driver is closed with respect to the standard bracket). This allows us to prove that the composition of a symmetry driver with the Fr\'echet derivative of an integral is also pre-Hamiltonian (in a new set of the variables) if both the symmetry driver and the integral have the smallest orders.

\end{minipage}

\bigskip

\newpage

 \section{Introduction: Pre-Hamiltonian operators}\label{preg}
 
Suppose that $u$ is a function of $x$ and $y$. Let ${\cal F}$ denote the set of all functions depending on a finite number of the variables
\begin{equation}\label{ydynvar}
y, \quad x, \quad u_0 = u, \quad u_1 = u_x, \quad u_2 = u_ {xx}, \quad \dots\,.
\end{equation}
Notice that \eqref{ydynvar} does not contain derivatives of $u$ with respect to $y$. Moreover, the objects defined in this Section do not employ the variable $y$, and this variable serves here as a parameter (which is absent in the standard definitions). The origin of the parameter $y$ will be explained in Section~\ref{hypsec}.  

By $D$ we denote the total derivative with respect to $x$. On functions from $\cal F$, it is defined by the formula
\begin{equation} \label{DxShort}
D = \frac{\pa}{\pa x}+\sum_{i = 0}^{+\infty} u_{i + 1} \frac{\partial}{\partial u_i}.
\end{equation}
%\begin{definition} 
For any function $a \in {\cal F}$, the differential operator
\begin{equation*}%\label{Frechet}
a_* \stackrel{def} {=} \sum_{i=0}^{+\infty} \frac{\partial a}{\partial u_i} D^i
\end{equation*}
is called the \emph{Fr\'echet derivative} of $a$.
%\end{definition}

For any $a \in {\cal F}$ we can also consider\footnote{In the formal sense, just as a way to define new differentiations on $\cal F$.} the total derivative $\partial_a$ with respect to $t$ in virtue (i.e. on solutions) of the evolution equation $u_t = a$. It easy to check that $\partial_a (h) = h_* (a)$ if $h \in {\cal F}$. The commutator $[\pa_f, \, \pa_g]$ of such derivatives corresponds to the Lie bracket
\begin {equation}\label{Lie}
\Big [f, \, \, g \Big] \stackrel{def}{=} g_* \, (f) -f_ * \, (g)
\end {equation}
on ${\cal F}$. Namely, it is well-known (see, for example,  \cite{Olv93}) that  
\begin{equation}\label{comm}
[\pa_f, \, \pa_g] = \pa_{[f,g]}.
\end{equation}

Let $M$ be a differential operator of the form
\begin{equation} \label{m1}
M=\sum_{i=0}^{k} \xi_i \,  D^i,  \qquad \xi_i \in {\cal F}, \qquad k \ge 0.
\end {equation}
The direct calculation shows that $D(a)_* = D \circ a_*$ for any $a \in {\cal F}$, where the symbol $\circ$ denotes the composition of differential operators. Taking this into account, we see that the relation 
\[ \Big[ M(a), M(b) \Big]= M\left( b_*(M(a))-a_*(M(b)) \right) + %\]\[
\sum_{i=0}^k \left( D^i(b) (\xi_i)_* (M(a)) - D^i(a) (\xi_i)_* (M(b))\right) \]
holds for any differential operator~\eqref{m1} and any functions $a,b \in {\cal F}$. Denoting the highest of the orders of the operators $(\xi_i)_* \circ M$ by $m$, we can rewrite this relation in the form 
\begin{equation}\label{sip}
\Big[ M(a), M(b) \Big] - M\left(b_* (M(a))-a_* (M(b))\right) = \sum_{i=0}^m \sum_{j=0}^m c_{ij}  D^i(a) D^j(b), 
\end{equation}
where the functions  $c_{ij}\in {\cal{F}}$ do not depend on $a$ and $b$. 

Generally speaking, the right-hand side of \eqref{sip} does not belong to the image of $M$ for all $a$ and $b$. But, for example, if $M=D+u_1$, then 
\begin{equation}\label{ldr}
\Big [M(a), \, M(b) \Big] = M \Big ( b_*(M(a)) - a_*(M(b))+D (a) \, b - D (b) \, a \Big)
\end{equation}
for any $a, b \in {\cal F}$. 

\begin {definition}
A differential operator $M$ of the form \eqref{m1} is called \emph{pre-Hamiltonian} if for any $a, b \in {\cal F}$ there exists $\vartheta \in {\cal F}$ such that $\Big [M(a), \, M(b) \Big] = M(\vartheta)$.
\end {definition}
It is known (see \cite {Olv93} for example) that the image $\im {\cal H}$ of an operator ${\cal H}$ satisfies the relation $\Big [\im {\cal H}, \, \im {\cal H} \Big] \subseteq \im {\cal H}$ (i.e., $\im {\cal H}$ forms a Lie subalgebra in ${\cal F}$) if the operator $\cal H$ is Hamiltonian. Thus, we can consider pre-Hamiltonian operators as a generalization of Hamiltonian operators. (Note that the pre-Hamiltonian operator $D+u_1$ in the above example is not Hamiltonian because it is not skew-symmetric.) 

Pre-Hamiltonian operators (under different names or without a name) were studied in several works\footnote{References to these works will be added in the next version of this text; %for now see references within 
a part of them are mentioned in subsection~3.5.2 of https://arxiv.org/abs/1711.10624} of this century.
As far as the author knows, the definition of pre-Hamiltonian operators actually arose for the first time in subsection~7.1 of \cite{zibsok} as an `experimentally observed' remarkable property of operators associated with hyperbolic equations of the Liouville type (i.e., with Darboux integrable equations).

The present text can be considered as an addition to these `experimental observations' of \cite{zibsok} and provides them with the proof. Namely, we prove in Section~\ref{phdarb} that any Darboux integrable equation generates four pre-Hamiltonian operators (a couple of such operators per each of the characteristics).

\section{Darboux integrable hyperbolic equations. Notation, definitions and known facts}\label{hypsec}

\subsection{Notation}\label{hnot}
From now on, we deal with hyperbolic equations of the form
\begin {equation}
\label {hyper} u_ {xy} = F (x, y, u, u_x, u_y).
\end {equation}
If $u(x,y)$ is a solution of equation (\ref{hyper}), then we can express all mixed derivatives of $u$ in terms of the variables 
\begin{equation}\label{hypdyn}
x, \quad y,\quad u_0 = \bar u_0 = u, \quad u_1 = u_x, \quad \bar u_1 = u_y, \quad u_2 = u_ {xx}, \quad \bar
u_2 = u_ {yy}, \dots \, \,.
\end{equation}
These variables will be considered as independent. We use notation $f\langle u \rangle$ for a function $f$ depending on a finite number of the variables (\ref{hypdyn}). We also use the notation of Section~\ref{preg} and, in particular, continue to write $f \in {\cal F}$ if $f\langle u \rangle$ does not depend on $\bar{u}_j$ for all $j>0$. 

For equations \eqref{hyper} there is a duality in notation between objects related to $x$ and $y$-characteristics. If we denote an $x$-object by a symbol, then the corresponding $y$-object is denoted by the same symbol with a dash above. For example, by $D$ we denote  the total derivative with respect to $x$ defined on  solutions of equation \eqref{hyper}, while the total derivative with respect to $y$ is denoted by $\bar D$. These derivatives  are defined by formulas
$$
D = \frac {\partial} {\partial x} + \sum_ {i = 0} ^ {\infty} u_ {i + 1}
\frac {\partial} {\partial u_i} + \sum_ {i = 1} ^ {\infty} \bar D ^ {i-1} (F)
\frac {\partial}
{\partial \bar u_i}, $$
$$ \bar D = \frac {\partial} {
\partial y} + \sum_ {i = 0} ^ {\infty} \bar u_ {i + 1} \frac {\partial} {\partial
\bar u_i} + \sum_ {i = 1} ^ {\infty} D ^ {i-1} (F) \frac {\partial} {
\partial u_i}. $$
Note that the restriction of $D$ onto ${\cal F}$ coincides with \eqref{DxShort}. 

\subsection{Darboux integrable hyperbolic equations}\label{darbsec}

\begin{definition}\label{defint} A function $W \langle u \rangle$ is called an $x$-{\it integral} of equation \eqref{hyper} if $\bar D(W)=0$.  Any function $ W (x) $ is called a {\it trivial} $x$-integral.  
\end{definition}
If we replace $x$ with $y$ (and $\bar D$ with $D$) in the above definition, then we obtain the definition of $y$-integrals\footnote{Notice that in different papers (even in articles of the same authors) the names for the integrals are varied: in a part of the papers, the term `$x$-integral' is used to denote $y$-integrals and vice versa.}. Using the symmetry of formula~\eqref{hyper} with respect to the interchange $x \leftrightarrow y$, we hereafter give only one of two `symmetric' definitions and statements.

Differentiating the defining relation $\bar D (W)=0$ with respect to the highest variable of the form $\bar u_i$, we obtain that any $x$-integral $ W $ does not depend on $\bar u_i$ for all $i>0$ and, hence, it has the form
$$ W = W (x, y, u, u_1, u_2, \ldots, u_p). $$
The number $p$ is called {\it order} of the $x$-integral $W$.

\begin{example} The function $\ds w=u_2- \frac {1} {2} u_1 ^ 2$ is an $x$-integral of the Liouville equation 
\begin{equation}\label{liou}
u_{xy}=\exp{u}.
\end{equation}
\end{example}

 Obviously, if $w$ is an $x$-integral of equation~\eqref{hyper}, then the expression
 \begin {equation} \label{genint1}
 W = S \left (x, w, D (w), \dots, D ^ j (w) \right)
 \end {equation}
 is also an $x$-integral of~\eqref{hyper} for any function $S$ and any $j \ge 0$.

\begin{proposition}\label{intst}{\rm \cite{zib}}
Any $x$-integral $W$ can be represented in the form \eqref{genint1}, where $w$ is an $x$-integral of the smallest order. In particular, $w=\phi (x, \tilde{w})$ if $\tilde{w}$ is another $x$-integral of the smallest order.
\end {proposition}

\begin {definition} \label {Darb} An equation of the form \eqref {hyper} is said to be {\it Darboux integrable} if this equation possesses both nontrivial $x$ and $y$-integrals.
\end {definition}
The Liouville equation~\eqref{liou} is the most known example of nonlinear Darboux integrable equation. 

\subsection{Hyperbolic equations of Liouville type }
Let us introduce the functions
 $$ H_1 \stackrel {def} {=} - D \left (\frac {\partial
F} {\partial u_1} \right) + \frac {\partial F} {\partial u_1}
\frac {\partial F} {\partial \bar u_1} + \frac {\partial F} {\partial
u}, \qquad
H_0 \stackrel {def} {=} - \bar D \left (\frac {\partial F} {\partial \bar
u_1} \right) + \frac {\partial F} {\partial u_1} \frac {\partial
F} {\partial \bar u_1} + \frac {\partial F} {\partial u}.
$$
Then we can define the functions $H_i$ for $i>1$ and $i<0$ by the recurrent formula
\begin {equation} \label {TodaH}
D \bar D (\log H_i) = - H_ {i + 1} -H_ {i-1} +2 H_i, \qquad i \in \Z.
\end {equation}
The functions $H_i$ are called \emph{Laplace invariants of equation} \eqref {hyper}.

The direct calculation shows that $H_0 = H_1 = \exp {u}$ and $H_2 = H_{-1} = 0$ for the Liouville equation \eqref{liou}.

\begin {definition}\label{LioInt} Equation \eqref{hyper} is called a \emph{Liouville-type equation} if $H_r = H _ {- s} \equiv 0$ for some integers $r \ge 1$ and $s \ge 0$.
\end {definition}

\begin {theorem}\label{led} {\rm \cite{zibsok1, andkam1}} An equation of the form~\eqref{hyper} is a Liouville-type equation if and only if this equation admits both a non-trivial $x$-integrals $W (x, y, u, u_1, u_2, \ldots, u_p)$ and a non-trivial $y$-integral $\bar W (x, y, u, \bar u_1, \bar u_2, \ldots, \bar u _ {\bar p})$. In addition, $s<p$ and $r \le \bar{p}$. 
\end {theorem}
Thus, Definitions \ref{LioInt} and \ref{Darb} are equivalent for scalar hyperbolic equations.

Differentiating the defining relations $D(W)=0$ and $\bar D (\bar W)=0$ with respect to $u_p$ and $\bar u_{\bar p}$, respectively, we obtain the following statement.

\begin{corollary}\label{psibarpsi} If \eqref{hyper} is a Liouville-type equation, then the relations
\begin{equation}\label{dpsi}
\frac {\partial F} {\partial u_1} = \bar D \, \log
\psi (x, y, u, u_1, \dots, u_ {p}),  \qquad
\frac {\partial F} {\partial \bar
u_1} = D \, \log \bar \psi (x, y, u, \bar u_1, \dots, \bar u _ {\bar p})
\end{equation}
hold for some functions $\psi$ and $\bar \psi$.
\end{corollary}

\section{Pre-Hamiltonian operators related to Liouville-type equations}\label{phdarb} 

\subsection{Symmetries}\label{symmetries} 

An inherent feature of Darboux integrable equations is the existence of higher symmetries in both $x$ and $y$-directions. 
\begin{definition} An equation of the form
\begin{equation}\label{xsym}
u_t = g \langle u \rangle 
\end{equation}
is a \emph{symmetry} of equation \eqref{hyper} if  $g$ satisfies the following relation
%\begin{equation}\label{symd}
$$\Big(D \bar D - \frac{\pa F}{\pa u_1} D - \frac{\pa F}{\pa \bar u_1} \bar D + \frac{\pa F}{\pa u} \Big) (g)=0.$$
%\end{equation}
If the right-hand side of a symmetry belongs to $\cal F$, then we call it an $x$-\emph{symmetry}\footnote{The right-hand side of $y$-symmetries depends on $x,y,u,\bar u_1, \bar u_2,...$}.
\end{definition}
We often identify symmetry \eqref{xsym} with its right-hand side $g$. The set of $x$-symmetries is a Lie algebra with respect to the bracket \eqref{Lie}.%Let us denote this Lie algebra  by $Sym_x$.

The following formula generates symmetries for Liouville-type equations (see, for example, Theorem 5 in \cite{zibsok}, where this formula was proved but was given in the slightly different form \eqref{MMbp}).
\begin {theorem}\label{invdr} Suppose that  for equation~\eqref{hyper} there exists a non-zero function $\psi \langle u \rangle \in \ker (\bar D - F_{u_1})$ and  $H_r=0$ for some $r>0.$ Let us define the operator $\cal M$ by the formula
\begin {equation} \label{MM}
{\cal M} = 
\begin{cases}
\frac{1}{H_1} (D - F_{\bar u_1}) \circ \frac{1}{H_2} (D - F_{\bar u_1}) \circ \dots \circ \frac{1}{H_{r-1}} (D - F_{\bar u_1}) \circ H_{r-1} \dots H_1 \psi &\text{if $r>1$,} \\
\psi & \text{if $r=1$.}
\end{cases}
\end {equation}
Then $u_t = {\cal M} (W)$ is a symmetry of the equation~\eqref{hyper} for any $x$-integral $W$.
\end {theorem}
The above Theorem is applicable to any Liouville-type equation because $H_r=0$ by Definition~\ref{LioInt} and the function $\psi$ exists by Corollary~\ref{psibarpsi}. In the case of Liouville-type equations (as well as in other cases when both relations~\eqref{dpsi} hold), we can rewrite~\eqref{MM} for $r>1$ in the form
\begin{equation}\label{MMbp}
{\cal M} =\frac {\bar \psi}{H_1} D \circ \frac 1 {H_2} D \circ \cdots \circ \frac 1 {H_ {r-1}} D \circ \frac {\psi H_1 \dots H_ {r-1}} {\bar \psi}
\end{equation}
found in \cite{zibsok}.  

Notice that Theorem~\ref{invdr} can be applied not only to Liouville-type equations but also to some equations~\eqref{hyper} without nontrivial $x$-integrals. In the latter case, $W$ is an arbitrary function of $x$.  

\begin{example}\label{dni}
We have $H_1=0$ and $\psi = u_1$ for any equation of the form $u_{xy}= \eta(y,u,u_y)\, u_x$. In this case, ${\cal M}$ is the operator of multiplication by $u_1$. This operator maps arbitrary function of $x$ to a symmetry of this equation, while the equation generically is not Darboux integrable (admits nontrivial $x$-integrals not for all $\eta$ in accordance with \cite{zib,star1}).
\end{example}

\begin{definition} An operator $M =\sum_{i=0}^{k} \xi_i \langle u \rangle D^i$, $\xi_k \ne 0$ is called an {\it $x$-symmetry driver} if $k \ge 0$ and $u_t = M (W)$ is a symmetry of equation \eqref{hyper} for any $x$-integral $W$. The integer $k$ is called {\it order} of the driver $M$.
\end{definition}
The above definition remains applicable even if \eqref{hyper} admits only trivial $x$-integrals (see Example~\ref{dni}). It is clear that for a driver $M$ the operator $M\circ D^i \circ W$ is a driver for any $i \ge 0$ and any $x$-integral $W$ (in particular, $W$ can be equal to $1$).

\begin{lemma}\label{dss} The coefficients of any $x$-symmetry driver do not depend on  $\bar{u}_j$ for all $j>0$. The leading coefficient belongs to the kernel of the operator $\bar D - F_{u_1}$.
\end{lemma}
\begin{proof} Let $M =\sum_{i=0}^{k} \xi_i \langle u \rangle D^i$ be an $x$-symmetry driver. Then 
\[ (D \bar D - F_{u_1} D - F_{\bar u_1} \bar D -F_u) (M(g))=0 \]
for arbitrary function $g(x)$. Collecting the coefficients of $g^{(i)}$ in this identity, we arrive at the following relations:
\begin{equation}\label{lattr}
\begin{split} %\begin{array}{rl}
(\bar D - F_{u_1})(\xi_k) & =0, \\[3mm]
(\bar D -F_{u_1}) (\xi_{i-1}) & = (F_u + F_{u_1} D + F_{\bar u_1} \bar D - D \bar D) (\xi_i), \quad 1 \le i \le k. 
\end{split} %\end{array}
\end{equation}
It follows from the first relation that $\xi_k$ does not depend on derivatives of $u$ with respect to $y$, while the second one implies that  $(\xi_{i-1})_{\bar{u}_j}=0$ for all  $j>0$ if  $\xi_i$ has the same property. 
\end{proof}

\begin{corollary}\label{minm} Let $M$ and $\tilde M$ be $x$-symmetry drivers of the smallest possible order for an equation of the form~\eqref{hyper}. Then $\tilde M = M \circ W$, where $W$ is an $x$-integral of~\eqref{hyper}.
\end{corollary}
\begin{proof} 
Let $M$ and $\tilde M$ have order $q$ and their leading coefficients be denoted by $\xi_q$ and $\tilde \xi_q$, respectively. Both $\bar D (\log \xi_q)$ and $\bar D(\log \tilde \xi_q)$ are equal to $F_{u_1}$ by Lemma~\ref{dss}. Therefore,  $\bar D (\log \tilde \xi_q - \log \xi_q) = 0$ and $\tilde \xi_q = W \xi_q$, where $W \in \ker \bar D$. This means that the operator $\tilde M - M \circ W$ has order less than $q$. 

On the other hand, the last operator maps $x$-integrals into symmetries. But \eqref{hyper} admits no $x$-symmetry drivers of order less than $q$ by the assumption of Corollary. Hence, all the coefficients of $\tilde M - M \circ W$ must be equal to zero.
\end{proof}

According to Theorem~\ref{invdr}, the operator ${\cal M}$ is an $x$-symmetry driver of order $r-1$. As it is demonstrated in \cite{star1}, $H_j=0$ for some positive $j \le r$ if \eqref{hyper} admits an $x$-symmetry driver of order $r-1$. This means that \eqref{hyper} admits no driver of order less than $r-1$ if $H_{r-1} \ne 0$. Therefore, ${\cal M}$ is a driver of the smallest order. In addition, relation \eqref{dpsi} defines the function $\psi$ up to multiplication by $x$-integrals. Hence, the operator $\cal M$ is defined up to transformations ${\cal M} \rightarrow {\cal M} \circ W$, where $W$ is an $x$-integral. Comparing this with Corollary~\ref{minm} and taking Lemma~\ref{dss} and Theorem~\ref{invdr} into account, 
we obtain the following statement.

\begin{proposition}\label{calmm}
Let equation~\eqref{hyper} admit $x$-symmetry drivers. Then $H_r=0$ for some $r>0$, the kernel of $\bar D - F_{u_1}$ contains non-zero elements and any $x$-symmetry driver of the smallest order is defined by the formula~\eqref{MM} with some $\psi \in \ker (\bar D - F_{u_1})$.
\end{proposition}

The equations of Liouville type also possess differential operators that map symmetries to integrals.

\begin {lemma}\label{lemcl} Let $W (x, y, u, u_1, u_2, \ldots, u_p)$ be an $x$-integral of equation \eqref {hyper}. Then the differential operator $W_*$ maps any symmetry to some $x$-integral.
\end {lemma}

The above Lemma directly follows from Lemma 1 in \cite{SokSt}. Not so formal, Lemma~\ref{lemcl} is valid because the total derivative with respect to $t$ in virtue of any symmetry $u_t = g$ and the operator $\bar D$ commute: $[ \pa_g , \bar D] = 0$.

\begin{corollary}\label{cdl} For any $x$-symmetry driver $M$ and any $x$-integral $W$ the operator
$$ L = W_* \circ M $$
maps $x$-integrals into $x$-integrals again. After rewriting $L$ in the form $\sum \mu_i D^i$,   the coefficients $\mu_i$ of this operator are $x$-integrals.
\end{corollary}

 The most significant operator of this kind is the composition
 \begin{equation} \label{Lop}
 \mathcal{L}= w_* \circ \mathcal{M},
 \end{equation}
 where $w$ is a nontrivial $x$-integral of the smallest order and the operator ${\cal M}$ is given by~\eqref{MM} (i.e., $\cal M$ is an $x$-symmetry driver of the smallest order by Proposition~\ref{calmm}).

\begin {example}\label{exli} Recall that $H_0=H_1 = \exp (u)$ and $H_{2}=0$ in the case of the Liouville equation~\eqref{liou}. In this case, we can set $\psi=1$ and the formula~\eqref{MM} gives us
\begin{equation}\label{mli}
{\cal M} = \exp (-u) D \circ \exp (u) = D + u_1.
\end{equation} 
Since $H_0$ does not vanish, $\ds  \, \, w = u_2 - \frac {1} {2} u_1 ^ 2$ is an $x$-integral of the smallest order (see the last sentence of Theorem~\ref{led}). We have $w_*=D^2 - u_1 D$ and
\begin{equation}\label{lioul} 
{\cal L} = (D^2 - u_1 D)\circ (D + u_1)   = D ^ 3 + 2 w D + D(w).
\end{equation}
In Section~\ref{preg} we show that the operator~\eqref{mli} is pre-Hamiltonian (see~\eqref{ldr}). As it is noted in \cite{zibsok}, the operator~\eqref{lioul} is also pre-Hamiltonian if we change the notation and rewrite this operator as $D^3 + 2 u D + u_1$.
\end{example}

\subsection{Main results}\label{results} 

\begin{theorem}\label{phm} %alt: Let equation~\eqref{hyper} admit $x$-symmetry drivers and the operator $\cal M$ be defined by~\eqref{MM} in accordance with Proposition~\ref{calmm}. 
Let $\cal M$ be an x-symmetry driver of the smallest order for equation \eqref{hyper}. 
Then the operator ${\cal M}$ is pre-Hamiltonian and, moreover, there exist functions $\gamma_{ij} \in \ker \bar D$ such that
\begin{equation}\label{phmdr}
\big[ {\cal M}(a),{\cal M}(b) \big] = {\cal M}\left( b_* ({\cal M}(a))-a_* ({\cal M}(b)) + \sum_{i=0}^n \sum_{j=0}^n \gamma_{ij}\, D^i(a) D^j(b) \right),  
\end{equation}
for any $a, b \in {\cal F}$. In particular, for any $x$-integrals $f$ and $g$ there exists $\phi \in \ker \bar D$ such that $\big[{\cal M}(f),{\cal M}(g)\big] = {\cal M}(\phi)$.
\end{theorem}
Note that $\cal M$ is always defined by formula~\eqref{MM} in accordance with Proposition~\ref{calmm}. %alt: To prove the above theorem, we use only the fact that $\cal M$ is an $x$-symmetry driver of smallest order (see Proposition~\ref{calmm}). 
It should also be emphasized that we do not assume the existence of nontrivial $x$-integrals in Theorem~\ref{phm}. If nontrivial $x$-integrals are absent for equation~\eqref{hyper}, then this theorem means that $\gamma_{ij}$ (as well as $f$, $g$ and $\phi$) are functions of $x$ only.

\begin{remark}[Preliminaries for the proof] 
We can rewrite equation~\eqref{phmdr} in the form 
\begin{equation}\label{br}
\sum_{i=0}^m \sum_{j=0}^m c_{ij} D^i(a) D^j(b) = {\cal M} \left( \sum_{i=0}^n \sum_{j=0}^n \gamma_{ij} D^i(a) D^j(b) \right), 
\end{equation}
where the functions $c_{ij}$ are defined by~\eqref{sip} for $\cal M$. The functions $a$ and $b$ are arbitrary here. Therefore, $D^i(a)$ and $D^j(b)$ 
can be considered as independent variables\footnote{I.e. a polynomial in these variables is equal to zero if and only if all the coefficients of this polynomial are zero.}. Collecting the coefficients at these variables in the right-hand side of~\eqref{br} and equating them to $c_{ij}$, we see that formula~\eqref{phmdr} is equivalent to relations between the coefficients of the operator $\cal M$ (which, in particular, determine $c_{ij}$) and the functions $\gamma_{ij}$. These relations do not contain $a$ and $b$ in any way and therefore remain unchanged if we prove~\eqref{br} for $a$ and $b$ belonging to an appropriate subset of $\cal F$ (such that $D^i(a)$ and $D^j(b)$ continue to play the role of independent variables when $a$ and $b$ are arbitrary elements of this subset). Thus, it is enough for the proof of Theorem~\ref{phm} to demonstrate that \eqref{phmdr} holds for the functions $a$ and $b$ of a special form. We take arbitrary $x$-integrals as such a special form for $a$ and $b$.
\end{remark}

For further reasoning, it is convenient to prove the following lemma first.

\begin{lemma}\label{red} Let $M$ %=\sum_{i=0}^{k} \xi_i\, D_x^i$
be an $x$-symmetry driver of order $k$ for equation~\eqref{hyper} and an expression  
$$\mathfrak{N}=\sum_{i=0}^q \sum_{j=0}^\ell c_{ij}  D^i(f) D^j(g), \quad c_{ij} \in {\cal F}, \quad q \ge k,
$$
be a symmetry of~\eqref{hyper} for arbitrary $x$-integrals $f$ and $g$. Then there exist $x$-integrals   $\theta_j$ such that the expression
%\begin{equation}\label{nev}
\[ \tilde{\mathfrak{N}} = \mathfrak{N} - M \left( D^{q-k} ( f \sum_{j=0}^{\ell}\theta_j D^j(g) ) \right) \]
%\end{equation}
has the form $\sum_{i=0}^{q-1} \sum_{j=0}^{\tilde{\ell}} \tilde{c}_{ij} \, D^i(f) D^j(g)$
and is a symmetry of~\eqref{hyper} for any $x$-integrals $f$ and $g$.
\end{lemma}
In simpler words, the driver $M$ allows us to reduce the order $q$ of $\mathfrak{N}$ 
without loosing the other property of $\mathfrak{N}$. Therefore, Lemma~\ref{red} remains applicable to the reduced expression $\tilde{\mathfrak{N}}$ if its order $\tilde q$ is greater than $k-1$.

\begin{proof}[Proof of Lemma.] The assumptions of the lemma imply that the operator $R = \sum_{i=0}^{q} \mathcal{C}_i(g) D^i$, where $\mathcal{C}_i(g)= \sum_{j=0}^\ell c_{ij} D^j(g)$, is an $x$-symmetry driver for any $g \in \ker \bar D$. According to Lemma~\ref{dss}, we have 
\[ (\bar D - F_{u_1})(\mathcal{C}_q(g))= \sum_{j=0}^\ell ( \bar D - F_{u_1})(c_{qj}) D^j(g) = 0 .\]
Since the integral $g$ is arbitrary, we get $(\bar D - F_{u_1})(c_{qj})=0$. The leading coefficient $\xi_k$ of the driver $M$ also belongs to the kernel of the operator 
$\bar D - F_{u_1}$. Therefore $c_{qj}=\theta_j \xi_k$, where $\theta_j \in \ker \bar D$ (see the proof of Corollary~\ref{minm} if a more detailed explanation is needed). The last equalities imply that
\[ M\left( D^{q-k}\left(f \sum_{j=0}^{\ell}\theta_j D^j(g)\right) \right) = \mathcal{C}_q(g) D^q(f) + \dots,\]
where the dots denote terms with  $D^i(f)$, $i < q$. Since the $M \circ D^{q-k}$ is a driver, the above expression is a symmetry for any $f, g \in \ker \bar D$. Subtracting this symmetry from  $\mathfrak{N}$, we complete the proof.  
\end{proof}

\begin{proof}[Proof of Theorem~\ref{phm}] As $\cal M$ is an $x$-symmetry driver, the functions ${\cal M}(g)$ and ${\cal M}(f)$ are symmetries of \eqref{hyper} for any $x$-integrals $f$ and $g$. Lemma~\ref{lemcl} implies that   $f_*({\cal M}(g))$ and $g_*({\cal M}(f))$ are $x$-integrals (in particular, are zero if \eqref{hyper} admits trivial $x$-integrals only). Therefore, ${\cal M}\Big( g_*({\cal M}(f)) - f_*({\cal M}(g))\Big)$ is a symmetry. In addition, $\big[{\cal M}(f),{\cal M}(g)\big]$ is also a symmetry since $x$-symmetries forms a Lie algebra with respect to the bracket \eqref{Lie}. Thus, substituting $\cal M$, $f$, $g$ for $M$, $a$, $b$ in formula~\eqref{sip}, we obtain that 
\begin{equation}\label{nstart}
\sum_{i=0}^m \sum_{j=0}^m c_{ij} D^i(f) D^j(g) 
\end{equation}
is an $x$-symmetry for any $x$-integrals $f$ and $g$.

%alt: According to Proposition~\ref{calmm},
By the assumption of Theorem~\ref{phm}, $\mathcal{M}$ is a driver of smallest order. Let us denote this order by $k$ (recall that $k=r-1$ in terms of Proposition~\ref{calmm}). Starting from~\eqref{nstart}, we apply Lemma~\ref{red} several times and arrive at the relation  
\[ [{\cal M}(f),{\cal M}(g)] - {\cal M}\left( g_* ({\cal M}(f))-f_* ({\cal M}(g)) + \sum_{i=0}^n \sum_{j=0}^n \gamma_{ij} D^i(f) D^j(g) \right) = \sum_{i=0}^{\hat{m}} \sum_{j=0}^{\hat{\ell}} \varsigma_{ij} \, D^i(f) D^j(g) , \] 
where $\gamma_{ij} \in \ker \bar D$ and the non-negative integer $\hat{m}<k$. If $k=0$, then the right-hand side of the above relation is zero because we can completely absorb \eqref{nstart} into the image of $\cal M$ by using Lemma~\ref{red}. If $k>0$ and there are non-zero coefficients $\varsigma_{ij}$, then the right-hand side defines a driver for arbitrary $g \in \ker \bar D$. But the order of this driver is less that the smallest order $k$ and, therefore, all the coefficients $\varsigma_{ij}$ must be equal to zero.\end{proof}

We have proved that the operator ${\cal M}$ is pre-Hamiltoninan with respect to the bracket \eqref{Lie} defined on functions from $\cal F$.

It turns out that there exists a pre-Hamiltonian operator in a set of variables other than~\eqref{ydynvar} if equation~\eqref{hyper} admits nontrivial $x$-integrals in addition to $x$-symmetry drivers.\footnote{As it is demonstrated in \cite{star1}, equation \eqref{hyper} is a Liouville-type equation if it admits both $x$-symmetry drivers and nontrivial $x$-integrals.} 

Let $w$ be an $x$-integral of smallest order (see Proposition \ref{intst}). 
Consider the operator ${\cal L}$ defined by \eqref{Lop}.  
According to Corollary \ref{cdl}, all coefficients of ${\cal L}$ are function of $x$, $w$, $w_i\stackrel{def}{=}D^i(w)$. 

\begin{theorem}\label{pgl} The operator ${\cal L}$ is pre-Hamiltonian with respect to the bracket 
\[ \{a,b\}= \sum_{i=0}^{+\infty} \left( \frac{\partial b}{\partial w_i} D^i(a) - \frac{\partial a}{\partial w_i} D^i(b) \right) \]
defined on functions of the variables $x$, $w$, $w_1$, $\dots$.
\end{theorem}

It was already noted in \cite{zibsok} that the operator $\cal L$ is pre-Hamiltonian for many (all checked) examples of Liouville-type equations. In other words, the work~\cite{zibsok}, in fact, contains the above theorem as a conjecture. By using Theorem~\ref{phm}, we can now prove Theorem~\ref{pgl}. The proof below belongs to V.~V.~Sokolov.

\begin{proof} Consider the evolution equations $u_{t_1}=\mathcal{M}(f)$ and $u_{t_2}=\mathcal{M}(g)$, where $f$ and $g$ are arbitrary function of the form~\eqref{genint1} (i.e. $f$ and $g$ are $x$-integrals). Recall (see Section~\ref{preg}) that differentiations with respect to $t_1$ and $t_2$ in virtue of these equation are denoted by $\partial_{\mathcal{M}(f)}$ and $\partial_{\mathcal{M}(g)}$. Let us apply the commutator $[\partial_{\mathcal{M}(f)},\partial_{\mathcal{M}(g)}]$ to the $x$-integral $w$ of smallest order. Taking \eqref{comm} and Theorem~\ref{phm} into account, we obtain
\[ [\partial_{\mathcal{M}(f)},\,\partial_{\mathcal{M}(g)}] (w)= w_* \left(\big[\mathcal{M}(f),\,\mathcal{M}(g)\big] \right) = w_* (\mathcal{M}(\phi)) = \mathcal{L}(\phi) \]
for some $x$-integral $\phi$ (which is a function of $x$, $w$, $w_i$ by Proposition~\ref{intst}). 

On the other hand, we have 
$$ w_{t_1} = \partial_{\mathcal{M}(f)}(w) = \mathcal{L}(f), \qquad  w_{t_2} = \partial_{\mathcal{M}(g)}(w) = \mathcal{L}(g).$$
This means that the differentiations $\partial_{\mathcal{M}(f)}$ and $\partial_{\mathcal{M}(g)}$ respectively coincide with $\partial_{\mathcal{L}(f)}$ and $\partial_{\mathcal{L}(g)}$ on functions of $x$, $w$, $w_i$. Since $\mathcal{L}(f)$ and $\mathcal{L}(g)$ are function of $x$, $w$, $w_i$ by Corollary~\ref{cdl}, we have
\[ [\partial_{\mathcal{M}(f)},\partial_{\mathcal{M}(g)}] (w) = \partial_{\mathcal{M}(f)} \left({\cal L} (g)\right) - \partial_{\mathcal{M}(g)} \left({\cal L} (f)\right) = \partial_{\mathcal{L}(f)} \left({\cal L} (g)\right) - \partial_{\mathcal{L}(g)} \left({\cal L} (f)\right) = \{\mathcal{L}(f),\mathcal{L}(g)\} . \]

Thus, we arrive at the relation
\begin{equation}\label{pregl}
\{\mathcal{L}(f),\,\mathcal{L}(g)\} = \mathcal{L}(\phi)
\end{equation}
that holds when $w_i$ are functions from $\cal F$. Since the functions $x$, $w$, $w_1$, $\dots$ $\in {\cal F}$ are functionally independent, the relation~\eqref{pregl} holds identically (i.e. without substituting $w_i$ with their values from $\cal F$). \end{proof}

\section*{Acknowledgements}
The author thanks V.~V.~Sokolov for many useful discussions as well as for suggestions that made %(as far as the author managed) 
the text more readable even in this very preliminary version.

\end{document}